\documentclass[journal]{IEEEtran}
\usepackage[cmex10]{amsmath}
\usepackage[utf8]{inputenc}
\usepackage{graphicx}
\usepackage{amssymb}
\usepackage{tikz-network}
\usepackage[english]{babel}
\usepackage{array}
\usepackage{textcomp}
\usepackage{cite}
\usepackage{hyperref}
\ifpdf
\hypersetup{
  colorlinks=true,
  linkcolor=blue,
  citecolor=blue,
  urlcolor=blue,
}
\fi

\usepackage{listings}
\usepackage{tikz-network}
\usepackage{amsthm}
\theoremstyle{definition}
\newtheorem{definition}{Definition}[section]

\usepackage{stfloats}

\newtheorem{theorem}{Theorem}[section]

\newcommand{\ie}{{\em i.e., }}

\newcommand{\etal}[1]{{\em et al.}~\cite{#1}}

\begin{document}
%
\title{The entropy rate of Linear Additive Markov Processes}
%
%
%

\author{Bridget~Smart, Matthew~Roughan and Lewis~Mitchell%
\thanks{B. Smart, M. Roughan and L.Mitchell are with the University of Adelaide. Email \{bridget.smart, matthew.roughan, lewis.mitchell\} @adelaide.edu.au}
\thanks{B. Smart would like to acknowledge the support of a Westpac Future Leaders Scholarship. M.~Roughan and L.~Mitchell are supported by the Australian Government through the Australian Research Council’s Discovery Projects funding scheme (project DP210103700).}%
}

\markboth{}
{Entropy of Linear Additive Markov Processes}

%

\maketitle
\section{Abstract}
This work derives a theoretical value for the entropy of a Linear Additive Markov Process (LAMP), an expressive model able to generate sequences with a given autocorrelation structure. While a first-order Markov Chain model generates new values by conditioning on the current state, the LAMP model takes the transition state from the sequence's history according to some distribution which does not have to be bounded. The LAMP model captures complex relationships and long-range dependencies in data with similar expressibility to a higher-order Markov process. While a higher-order Markov process has a polynomial parameter space, a LAMP model is characterised only by a probability distribution and the transition matrix of an underlying first-order Markov Chain. We prove that the theoretical entropy rate of a LAMP is equivalent to the theoretical entropy rate of the underlying first-order Markov Chain. This surprising result is explained by the randomness introduced by the random process which selects the LAMP transitioning state, and provides a tool to model complex dependencies in data while retaining useful theoretical results. We use the LAMP model to estimate the entropy rate of the LastFM, BrightKite, Wikispeedia and Reuters-21578 datasets. We compare estimates calculated using frequency probability estimates, a first-order Markov model and the LAMP model, and consider two approaches to ensuring the transition matrix is irreducible. In most cases the LAMP entropy rates are lower than those of the alternatives, suggesting that LAMP model is better at accommodating structural dependencies in the processes.

 \begin{IEEEkeywords}
LAMP processes, entropy, Markov processes, long-range dependency, stochastic processes.
\end{IEEEkeywords}

\section{Introduction}
 
Markov processes are a simple model with wide-ranging applications. They rely on the Markov property, where the next state is only dependent on the current state, and is conditionally independent of the history of the process. The model's simplicity allows theoretical results to be easily derived and understood, and the flexibility of the model's framework makes it a useful tool which can be easily extended, and has consequently been used in a vast number of applications. 

However, many real sequences have correlations with a longer or more complex history than the simple Markov property allows~\cite{leland94:_self_simil_natur_ether_traff_exten_version , Beran}. It is not uncommon to find correlations in real data that are so strong that the process is considered to have long-range dependence, which is defined by the tail of the autocorrelation function decaying so slowly that sums over the tail diverge. Other processes with only short-range correlations may still exhibit complex structure of interest. 

Markov processes do have a correlation structure, and this can even be long-range dependent for infinite-state Markov chains \cite{carpio2007long}, but the autocorrelation structure of a Markov process is difficult to tune, reducing their usefulness in applications that require control over the autocorrelation structure. 

Higher-order Markov processes allow the next transition to depend on the previous $m$ states, where $m$ describes the order of the process, and thus allow longer correlations to be built into the process. For a $m$th-order Markov Process with $n$ possible states, modelling these transitions requires $(n-1) n^m$ parameters, so higher-order models are prone to overfitting as $m$ increases. 

The Linear Additive Markov Process (LAMP), proposed by Kumar~\etal{10.1145/3038912.3052644}, overcomes these challenges. The LAMP can fit a measured autocorrelation structure even for
sequences that exhibit long-range dependence. An expressive model, the LAMP model has been shown to have comparable performance to deep sequential models despite their relatively small parameter space \cite{10.1145/3038912.3052644}. While LAMPs are not as expressive as higher-order Markov processes, their straightforward theoretical results and explainability make them a good model for time-series data where the measured correlation structure is of interest.

This work proves that the entropy rate of a LAMP is equal to that of the underlying first-order Markov Chain. The result is surprising because the structure introduced by the LAMP plays no part in entropy rate, despite having a strong impact on the structure of the process, for instance, through the autocorrelation. 
A consequence of this result is that the LAMP is able to distinguish between sequences that might look very different due to the presence of long-range dependency (LRD), yet have the same entropy rate -- providing a further useful application of the model.

In summary, the contributions of the paper are:
\begin{enumerate}
    \item A closed form solution giving the entropy rate of a LAMP process, and showing that the entropy rate is exactly that of the underlying Markov chain regardless of the autocorrelation.
    
    \item Use of this result to construct an estimator which is applied to four publicly available datasets to obtain entropy estimates, close to or lower than values obtained using a first-order Markov model. Estimated entropy values are 4.88 bits/symbol for the Lastfm dataset, 2.49 bits/symbol for the BrightKite dataset, 3.18 bits/symbol for the Wikispeedia dataset and 3.61 bits/symbol for the Reuters dataset.
    
    \item An exploration of the impacts of forcing ergodicity in a transition matrix on entropy estimates. This work considers two approaches, using only the largest connected component and adding an artificial fully connected node, with entropy estimates similar between each approach. Analysis into the stability of the estimate against the artificial transition probability was undertaken, and a robust value of $p_{\text{artificial}} = 2^{-15}$ is recommended.

\end{enumerate}

\section{Background}

We start by defining a standard discrete-time, time-homogeneous Markov chain $X_t$, a process which underlies the LAMP model. The definitions and notation are standard, but we include them to ensure there is no confusion. 

Let $S$ be a state space, which contains $n$ states. Without loss of generality, we label the states $S=[n]=\{1,2,\ldots, n\}$. We take $P$ to be the $n \times n$ stochastic matrix which defines the transitions between states for the underlying first-order Markov process. It follows that the row sums of $P$ are all equal to 1, and each element is non-negative.

In a first-order Markov process, the probability of transitioning from state $x_{t-1}$ to $x_t$ is given by,
\begin{align*}
&\Pr(X_n=x_n | X_0=x_{0},...,X_{n-1}=x_{n-1}) \\
&= \Pr(X_n=x_n |X_{n-1}=x_{n-1})&\\ &= P_{x_{n-1}, x_{n}},
\end{align*}
by the Markov property. We also have that the stationary distribution of the Markov process is denoted by the vector $\pi$, which is the left eigenvector of $P$ corresponding to the eigenvalue of 1. That is, $\pi$ satisfies the equation
$$
\pi P = \pi.
$$

If $P$ defines an ergodic Markov process, this stationary distribution is unique.

The above results can be generalised to countable state processes, however, a finite-state Markov chain cannot reproduce long-range dependency structures. LAMPs afford control over the correlations within the process without requiring an infinite state space.

\subsection{Linear Additive Markov Process (LAMP)}

A Linear Additive Markov Process (LAMP) is defined using a first-order Markov process and a probability distribution $w$ on the positive integers that we refer to as the LAMP kernel. Assuming we are at time $n-1$ and wish to calculate the next state $X_{n}$ we adopt the following two stage process:
\begin{enumerate}
    \item Select a past state $x_{n-q}$ with probability $w_q$. This state is also referred to as the \textit{transition state}, $X_{T_n}$.
    
    \item Determine the state $X_n=x_n$ according to the probabilities $P_{x_{n-q}, x_{n}}$.
    
\end{enumerate}
That is, we use the same state transition probabilities as in the underlying Markov chain, but the historical state on which the transition is based is randomly chosen according to $w$. 

In practice, there are some extra details: (i) the distribution $w$ is typically assumed to have finite support $[k] = \{1,2,\ldots, k\}$ but it need not be finite, and (ii) the state that is chosen must be from the existing history, which does not extend $k$ steps initially, and so the state chosen as the jumping off point is masked according to $x_{\max\{0,n-q\}}$. The value of $k$ gives the {\em order} of the LAMP.

An alternative equivalent definition is provided in \cite{10.1145/3038912.3052644} as follows:

\begin{definition}[\textbf{LAMP Transitions}]
Given a stochastic matrix $P$ and a distribution $w$ on $[k]$, the $k$th-order LAMP evolves according to the following transition probabilities:
$$
Pr( X_n=x_n| x_0,...,x_{n-1} ) = \sum_{q=1}^k w_q P_{x_{\max\{0,n-q\}},x_n}
$$
\end{definition}

In their defining paper Kumar~\etal{10.1145/3038912.3052644} provide a number of important results:
\begin{enumerate}
    \item There exists a LAMP of order $k$ that cannot be approximated with any constant factor by a higher order Markov process of order $m-1$. So the model has more expressive power than a conventional Markov chain, with an exponentially smaller number of parameters. 

    \item Under mild regularity conditions (\ie Ergodicity of the underlying Markov chain) LAMPs possess a limiting equilibrium distribution which is the same as that of the underlying simple Markov chain, \ie the stationary distribution $\pi$. 
    
    \item The choice of the kernel $w$ determines how the history is included into the next state, and so can incorporate correlations extending as far back as desired, including long-range correlations (for $k \rightarrow \infty$). It is not stated, but it is plausible that this is closely related to the convolution of the conventional autocorrelation of the underlying Markov chain and the LAMP kernel $w$. We show this in action in some later examples.
    
    \item  Kumar~\etal{10.1145/3038912.3052644} provide a measure of the speed of convergence to the equilibrium distribution, though interestingly these results only apply where $w$ has a finite fourth moment, and so these convergence estimates do not apply in the case of a long-range dependent process. 
    
\end{enumerate}

The authors also provide an approximate Maximum Likelihood Estimator (MLE) to learn a LAMP from data using alternating estimation of $w$ and $P$. 

Thus LAMPs are a simple and expressive model,  characterised by a 
(first-order Markov) transition matrix and a discrete probability distribution on the non-negative integers. Their efficient parameterisation and simple theoretical results mean they can be easily learned and fit to data.

Kumar~\etal{10.1145/3038912.3052644} provide many useful results for their model but do not provide an explicit formula for the entropy rate of the process, which is the topic of this paper. 

\subsection{Entropy and Entropy Rate}

The entropy rate of a sequence is a quantity that describes the self-information the sequence contains. It also represents an asymptotic lower bound on the lossless compression ratio of the sequence and can be used to measure how predictable a sequence is \cite{CoverThomas}. 

For a discrete random variable, $X$, the Shannon entropy is denoted $H$, and is defined to be
$$
H(X) = -\sum_{x} p(x) \log p(x),
$$
where $p(\cdot)$ is the probability distribution of $X$, and with the convention that $0 \log 0 = 0$.
The joint entropy of $n$ random variables $X_i$ is the obvious extension:
\begin{flalign*}
    &H(X_1, \ldots, X_n) \\
    &= -\sum_{x_1,\ldots, x_n} p(x_1, \ldots, x_n)   \log p(x_1,\ldots, x_n).
\end{flalign*}

The {\em entropy rate} for a process ${\mathcal X}$ is just the limiting average of the joint entropy, 
$$
H({\mathcal X}) = \lim_{n \rightarrow \infty} \frac{1}{n} H(X_1, \ldots, X_n).
$$

A standard result \cite[Thm 4.2.1]{CoverThomas} links the entropy rate to an equivalent defined in terms of conditional probabilities,
$$
H'({\mathcal X}) = \lim_{n \rightarrow \infty} H(X_n | X_1, \ldots, X_{n-1}).
$$
For a stationary process the limits $H$ and $H'$ exist and are equal~\cite[Thm~4.2.1]{CoverThomas}, \ie $H({\mathcal X}) = H'({\mathcal X})$. That is, for a stationary stochastic process, the conditional entropy rate and per symbol entropy rate of $n$ random variables are equal in the limit. 

The result is convenient, particularly for computing the entropy rate of a stationary Markov chain, which is given by \cite[Thm 4.2.4]{CoverThomas} as
\begin{equation*}
    H({\mathcal X}) = - \sum_{ij} \pi_i P_{ij} \log P_{ij},
\end{equation*}
where $P$ is the probability transition matrix, and $\pi$ the stationary distribution. The primary result of this paper is that this formula also provides the entropy rate of a LAMP. 

\section{The entropy rate of a LAMP}

\begin{theorem}[\textbf{Entropy rate of a LAMP}]
For a LAMP defined by an underlying stationary, first-order Markov Chain with transition matrix $P$, and kernel distribution $w$ on $[k]$, the entropy rate of the LAMP is
\begin{equation}
    H({\mathcal X}) = - \sum_{ij} \pi_i P_{ij} \log P_{ij},
    \label{eqn:markov_entropy}
\end{equation}
where $\pi$ is the stationary distribution of the Markov chain.
\end{theorem}


\begin{proof}
We begin by considering the entropy rate of $n$ realisations from a LAMP, using known results from information theory about conditional entropy and the entropy rate of a first-order Markov Chain. As noted above, for a stationary Markov chain
\[ H({\mathcal X}) = H'({\mathcal X}). \]
Assume that at time $n-1$ we choose $X_{T_n} = X_{n-q}$ to be the transition state according to distribution $w$. The state at time $n$ depends only on $X_{n-q}$ and so the conditional entropy 
\[  H(X_n | X_{n-1},...,X_0,q) 
    = H(X_n | X_{n-q}, q ),
\]
where we include $q$ explicitly in the conditioned term for notational convenience. A transition from state $X_{n-q} \rightarrow X_n$ in the LAMP is governed by the same probability transition matrix $P$ as the transition $X'_{n-1} \rightarrow X'_n$  in the underlying Markov chain (we use $X'$ here to denote the underlying Markov chain). Furthermore \cite{10.1145/3038912.3052644} shows that if the Markov chain corresponding to $P$ is ergodic, then its stationary distribution $\pi'$ is also the stationary distribution of the LAMP, so we denote both by $\pi$. Hence, the conditional entropy is the entropy rate of the Markov chain, \ie 
\begin{eqnarray*}
  H(X_n | X_{n-q}, q ) 
   & = & H(X'_n | X'_{n-1} ) \\
   & = & - \sum_{ij} \pi_i P_{ij} \log P_{ij}.
\end{eqnarray*}
by \cite[Thm~4.2.4]{CoverThomas}. Note that this decouples the entropy rate from the choice of $q$ (assuming $n$ is large enough that $n-q \geq 0$ for all choices of $q$, or alternatively presuming an infinite history of the stationary process). 

Now, we can remove the conditioning on $q$ by noting that the overall conditional entropy is a probabilistically weighted sum over the choices of $q$, \ie
\begin{eqnarray*}
    \lefteqn{H(X_n | X_{n-1},...,X_0)} \\
    & = &  \sum_k  P(T_n = n-q)  H(X_n|X_{n-q},q) \\
    & = &  - \sum_{ij} \pi_i P_{ij} \log P_{ij} \sum_q w_q.
\end{eqnarray*}
Noting that $w$ is a proper probability distribution, its sum is 1, and hence the conditional entropy is given by the constant term, and hence the limit is trivial and the result follows. 
$$
H(\chi_{LAMP}) = H(\chi),
$$
where $H(\chi)$ is the entropy of the underlying first-order Markov Chain.
\end{proof}




\begin{figure*}
    \centering
    \includegraphics[width=\textwidth]{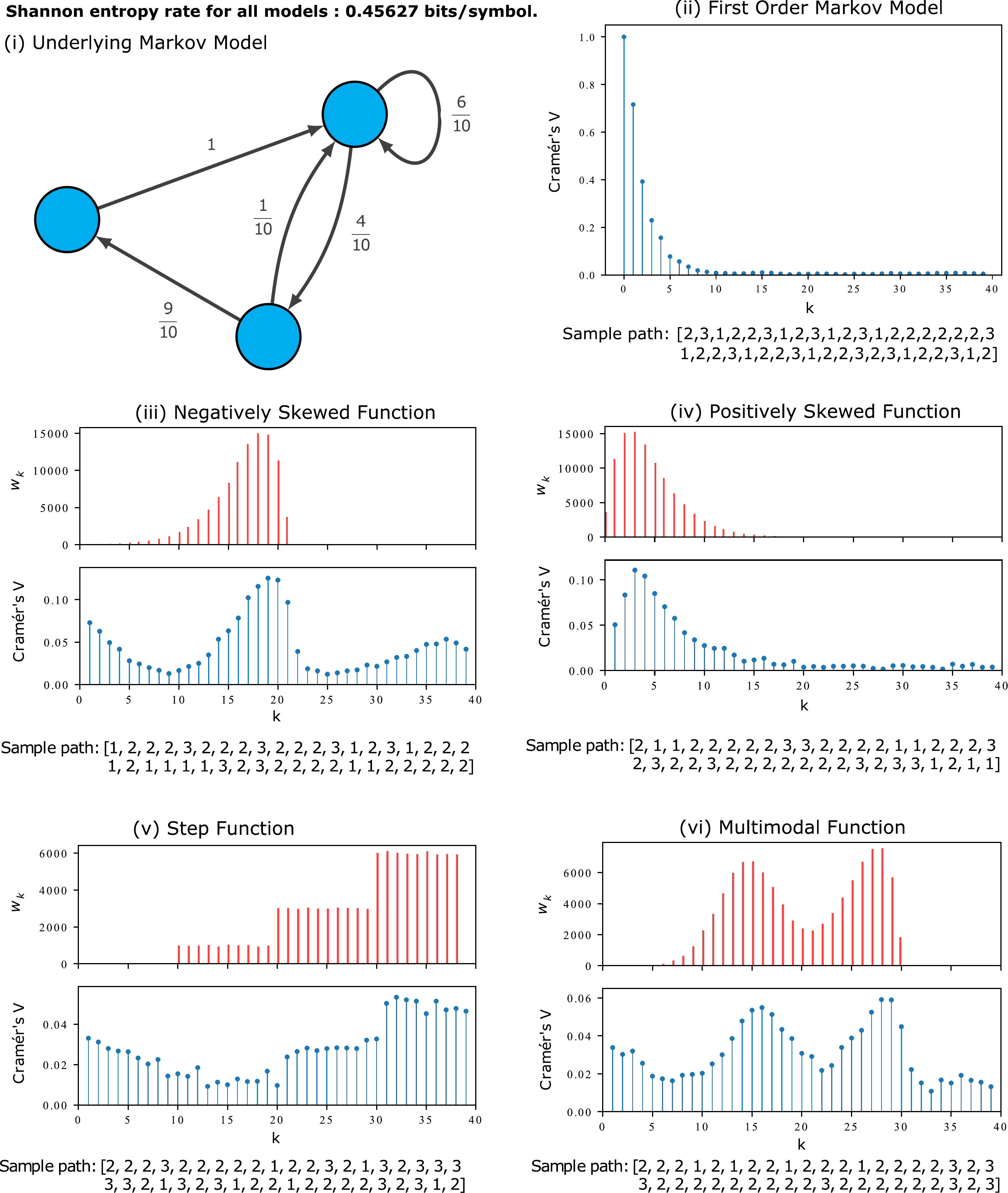}
    \caption{(i) A diagram of the underlying transition matrix for each of the models shown in (ii) - (vi). While these models all produce sample paths with differing dependency structures, they all have the same entropy rate, as proven in Theorem 4.1. In blue, the dependency structure of each model is visualised using Cramér's V for lags up to 40. This provides a measure of correlation for symbolic variables and is analogous to considering a lagged-cross correlation function. In (ii), the dependency structure is shown for a first-order Markov model. In (iii) - (vi) we see the frequency of $w_k$, or the distribution of the size of the backward step directly influences the shape of the dependency structure. While the Cramér's V is calculated for a lag of 0 for the first-order Markov Model, it is excluded from plots (iii)-(vi) since the value is much larger due to the Markov Chain's memoryless property. It was excluded to allow the similarity between the frequency distribution of step size and dependency structure to be more easily seen.}

    \label{fig:w_various}
\end{figure*}

The most noteworthy feature of this result is that LAMP kernel $w$ plays no part in the formula. This is surprising because the kernel clearly affects the autocorrelation structure of the process. Naive intuition suggests that a process with stronger autocorrelations should present less information at each time step, and hence have a smaller entropy rate.

While at first surprising, this result is intuitive as once the conditioning state, $X_{T_n}$ is selected, the process transitions in the same manner as a first-order Markov Chain. The randomness introduced by the choice of $k$, balances with the increased predictability which correlation with the history of the process provides.

Alternatively, we might think of this as increasing the correlations in the process, while also increasing the variance per time step, maintaining a constant entropy rate. However, we work here in the context of a finite alphabet of (not necessarily ordinal) states, so the common notions of variance have been subtly generalised by the LAMP model. 

This useful result provides a model able to capture complex dependency structures in data but maintain consistent randomness on a symbol by symbol level.

\section{Model Comparisons}

To demonstrate the significance of the LAMP model, we will consider the dependency structure of 5 comparable models: a first-order Markov model, and  four alternative LAMP models with differing $w_k$ (\autoref{fig:w_various}). Each of these models share the same transition matrix illustrated in \autoref{fig:w_various}~(i). Since they share a transition matrix, these sequences all share the same stationary distribution and entropy rate of 0.45627 bits per symbol, so each model is superficially similar.

The dependency structure is measured by calculating Cramér's~V statistic for various lags. Based on Pearson's Chi-Squared test, the Cramér's~V statistic measures the association between two discrete, nominal variables. To measure the self association for a sequence with a given lag $l$, pairs of values from the sequence at index $i$ and $i+l$ are used to construct a contingency table, which is used to calculate Cramér's~V statistic. The statistic ranges from 0 to 1, where 0 indicates no correlation, and 1 indicates a perfect relationship between the pairs of variables. Calculating this value for various values of $l$ produces a measure of dependency structure comparable to a lagged cross-correlation plot, but without assigning ordinal value to the state labels. 

To illustrate the autocorrelation structure for a standard first-order Markov model, it is shown in plot \autoref{fig:w_various}~(ii) along with a short sample path. \autoref{fig:w_various}~(iii)-(vi) show the Cramér's V Statistic for the LAMP models in comparison to the choice of $w$. Note that lag = 0 is not plotted in these cases, as this value will always be 1.

It is evident that the dependency structure mirrors the shape of $w_k$, demonstrating that the LAMP framework provides a useful model for generating sequences with a desired dependency structure, although it is important to note this structure will repeat and diminish for integer multiples of the original lag, as well as including structure from the original Markov model, so the two are not identical. However, there is an evident ability to tune the autocorrelation of a LAMP, while preserving other properties such as the stationary distribution and entropy rate. 

\section{Empirical Evaluation}

Following the same approach as the original LAMP paper~\cite{10.1145/3038912.3052644}, we fit and estimate Shannon Entropy rates on four publicly available datasets, \textsc{Lastfm}, \textsc{BrightKite}, \textsc{Wikispeedia} and \textsc{Reuters}. 
Each dataset is comprised of a number of distinct sequences which represent the activity of an individual or a single process. 
We compare the entropy rate estimates of (1) an empirical estimator, (2) the estimator obtained by fitting a Markov chain, and (3) the entropy estimate obtained by fitting a LAMP.

Our code is available at \url{github.com/bridget-smart/LAMPEntropyEstimates} and the links to each dataset are provided in the text below.

The \textsc{Lastfm} dataset contains the user activity from the music streaming service \url{last.fm} for 992 users. Users can choose to listen to stations based on a genre or artist, a particular song, and can share their listening activity. The dataset contains the user, timestamp, artist and song, but here we construct sequences using only the user and artist. Each sequence contains data for a single user, with items representing different artists. Containing 19M items, the data is available at \url{http://ocelma.net/MusicRecommendationDataset/lastfm-1K.html}.

The \textsc{BrightKite} dataset contains check-in data from the location-based social networking service BrightKite, where users share their checked-in locations. Each sequence represents a single user, with items representing the locations which the user has checked into. 772,967 locations and 51,406 users are represented in the dataset. The data is available at \url{https://snap.stanford.edu/data/loc-brightkite.html}.

The \textsc{Wikispeedia} navigation path dataset is comprised of sequences representing a path along Wikipedia links from a condensed version of Wikipedia. These sequences are collected using the online game \textsc{Wikispedia} where users race to navigate between articles. This dataset contains 51,318 complete paths. The dataset is available at \url{https://snap.stanford.edu/data/wikispeedia.html}.

For the \textsc{Reuters} dataset we use the Reuters-21578, Distribution 1.0 benchmark corpus, containing newswire articles, with each sequence representing a single article and items representing distinct words. This dataset is available at \url{www.nltk.org/nltk_data}.

Following the same process as in \cite{10.1145/3038912.3052644}, consecutive repeated values within each path repeated values are removed to prevent self-loops from appearing in the transition matrices. Consecutive repeated values are also not meaningful in our datasets, as they do not contain temporal data. We also replace values which appear less than 10 times, or less than 50 times for the \textsc{Lastfm} dataset with a unique token, representing a less frequently visited state. This preprocessing ensured our model fit was consistent with the original paper, and across both the LAMP and first-order Markov models.

LAMP models are fit to each dataset using the original LAMP code available at \url{github.com/google-research/google-research/tree/master/lamp}. Modified code used to obtain the LAMP transition matrices for this paper is available at \url{github.com/bridget-smart/modified_lamp}. The code used to obtain the entropy estimates is available at \\\url{github.com/bridget-smart/LAMPEntropyEstimates}.

While a modified version of the author's original code was used to fit the LAMP model, care was taken to exactly reproduce the preprocessing steps outlined in Kumar~\etal{10.1145/3038912.3052644} when fitting the remaining models. The code from Kumar~\etal{10.1145/3038912.3052644} fits the LAMP model using a subset of the data and a cross-validation technique to overcome computational and time constraints. Therefore, the estimate using the LAMP model was fit on a smaller subset of the data than other estimates. Discrepancies in the transition matrix size between the first-order Markov and LAMP models indicate the differences caused by this. 

Stationary distribution derivations are valid where the transition matrix is irreducible, but the data is imperfect, and this condition is not true for all transition matrices estimated from these datasets. We overcome this in two ways: (1) by only considering the largest strongly-connected component of the graph generated by considering the transition matrix as a directed network; and (2) by artificially connecting each distinct communicating class to the first state with transition weight $p_{\text{artificial}}$. 

The first approach assumes that the largest connected component of the transition matrix is representative of the chain's behaviour. To assess this we consider how many nodes are excluded from that component. For the first-order Markov models, the number of nodes ignored when only considering the largest fully connected component were 1, 1255, 68, 2 for each of the four datasets, \textsc{Lastfm}, \textsc{BrightKite}, \textsc{Wikispeedia} and \textsc{Reuters} respectively. The total number of nodes in each of these models were 23870, 43596, 3425 and 8465  respectively. The \textsc{BrightKite} dataset has the largest proportion of nodes removed (0.02879), which suggests the component can be representative.

In the LAMP models 237 (of 17766 for \textsc{Lastfm}), 2591 (of 25348 for \textsc{BrightKite}), 2775 (of 3504 for \textsc{Wikispeedia}) and 3094 (of 6822 for \textsc{Reuters}) nodes were ignored while only considering the largest fully connected component to obtain the entropy estimates for the LAMP models. These proportions are generally higher than those for the first-order Markov model, ranging from 0.7920 to 0.01334, which may result in more instability within the estimated entropy values.

The second approach forces irreducibility by ensuring that the chain mixes, and all states can be visited. It, and in particular a discussion of how to choose $p_{\text{artificial}}$ is given in Appendix~A, but a brief description is given here. This approach uses the assumption that it is possible to communicate between each possible pair of states, an assumption which is reasonable for most of our datasets, as, for instance, it is possible for users to listen to any pair of artists, but rare transitions are not likely to be observed, and so including a small transition probability to connect the network is reasonable. 
 Artificially inducing transitions between states which are otherwise disconnected, modifies the behaviour of the chain.  It follows that these artificial links should have a small weight since these transitions have not been observed within the data sample suggesting they are rare, but the weight of these links needs to be large enough to ensure mixing can occur, stable estimates can be calculated and errors introduced by computer precision are avoided. The choice of $p=2^{-15}$ and  $p=2^{-10}$ is discussed in Appendix~A.

\begin{table*}[t]
\begin{center}
\begin{tabular}{ >{\raggedright}r|r|rrrr}
                      & \multicolumn{4}{|c}{Data Set} \\
Estimate Type & Subtype & \textsc{Lastfm} & \textsc{BrightKite} & \textsc{Wikispeedia} & \textsc{Reuters} \\
\hline
 Shannon Empirical & Sequence Level & 11.83825 & 9.83813 & 9.97541 & 9.35027 \\
 Shannon Empirical & Path Level & 6.77366 & 0.88488 & 2.53676 & 5.13913 \\ 
 Shannon Empirical & Stationary Distribution & 8.20559 & 6.72001 & 6.67657 & 6.47187 \\ 
First-order Markov & Largest CC & 3.69304 & 1.67316 & 3.14187 & 4.09229\\  
 First-order Markov & Induced Irreducibility\footnotemark[1] & 3.69334 & 1.58479\footnotemark[2] & 3.93053 & 4.50134 \\ 
 LAMP & Largest CC & 4.88141 & 2.49536 & 3.18030 & 3.61305 \\
 LAMP & Induced Irreducibility\footnotemark[3]
  & 4.88092 & 2.53492 & 3.18052 & 3.61328 \\  

\end{tabular}
\caption{Entropy estimates for each of the four datasets using seven estimation methods. Two main approaches are used to ensure transition matrices are ergodic: only considering the largest connected component; and adding an artificial state with transition weight $p_{\text{artificial}}$. Both of these approaches give similar results for each of the four datasets. \label{tab:res_tab}}
\end{center}
\end{table*}

\footnotetext[1]{Values calculated using an artificial edge weight of $2^{-15}$, except for the value for the \textsc{Brightkite} dataset.}

\footnotetext[2]{This value was calculated using an articifical edge weight of $2^{-10}$, rather than the value of $2^{-15}$ used for all the other estimates. For more details see Appendix A.}

 \footnotetext[3]{Values calculated using an artificial edge weight of $2^{-15}$.}

The results are shown in \autoref{tab:res_tab}. The estimators referred to in the table are defined in detail below. At a high level:
\begin{itemize}
    \item \textit{Shannon Empirical} estimators use \begin{align*}
     H(\chi ) 
     & = - \sum_{x} \hat{p}(x) \log \hat{p}(x), 
    \end{align*}
    where $\hat{p}(x)$ is an estimated probability distribution. We consider three cases, defined below. 
    
    \item Markov estimates use the standard Markov-chain entropy given in \autoref{eqn:markov_entropy}, with empirical estimates of the transition and stationary probabilities. They do so either on the Largest Connected Component (CC) or on a chain with artificially induced irreducibility as described earlier. 
    
    \item LAMP estimates use the LAMP entropy, which is also given in \autoref{eqn:markov_entropy}, with empirical estimates of the transition and stationary probabilities, but note that the transition matrices for LAMP models will differ from those of the Markov case. These estimators also use either on the Largest Connected Component (CC) or on a chain with artificially induced irreducibility as described earlier. 
    
\end{itemize}

Shannon estimators are applied to a probability distribution, but the structure of the data provides several alternative distributions to consider. 

The data used contains multiple sample paths, \ie we have a set of $n$ subsequences $\{ y_i \}_{i=1}^n$, each composed of $m_i$ items, $ \{ s_{ij} \}_{j=1}^{m_i}$ for a total of $N$ realisations across all subsequences. 

The \textit{Sequence Level} estimator uses the frequency of each item across all subsequences to estimate its probability of occurrence, \ie
$$ \hat p(x) = \dfrac{1}{N} \sum_{i=1}^n\sum_{j=1}^{m_i} \mathbb{I}(x = s_{ij}) ,$$
where $\mathbb{I}(x = s_{ij})$ is an indicator variable which takes the value 1 if the $x = y_{i}$.

The \textit{Path Level} estimator uses the frequency of each symbol in a sequence to estimate a simple Shannon entropy for each sequence, and then averages over these, \ie
\[ \hat p_i(x) = \dfrac{1}{m_i} \sum_{j=1}^{m_i} \mathbb{I}(x = s_{ij}) ,\]
where $\mathbb{I}(x = s_{ij})$ is an indicator variable which takes the value 1 if $x = s_{ij}$, and 
\[ \hat{H}_i ( y_i )
    =  \sum_{x} \hat{p}_i(x) \log \hat{p}_i(x),
\]     
and
\[ \hat{H}( \chi )
    =  \frac{1}{n} \sum_{i=1}^n \hat{H}_i (y_i ),
\]
 

The \textit{Stationary Distribution} estimator uses as its probability distribution the stationary distribution as estimated through one of the Markov models. 


The results shown in \autoref{tab:res_tab} provide several insights:
\begin{enumerate}
    \item The Shannon sequence-level estimator grossly overestimates entropy values for all cases. The most likely scenario is that there just are not enough sampled sequences to properly represent the underlying distribution, so the calculation is essentially computing the entropy of a distribution close to the uniform distribution, \ie near the maximum possible value.
    
    \item The Shannon entropy of the stationary distribution is also a fair over-estimate indicating that the transition structure is important in these sequences. 
    
    \item The Shannon path-level estimator produces some overestimates, but also a very low value for the {\tt BrightKite} data. This anomalous result may arise here for several reasons, but it is noteworthy that the Shannon estimator can be biased on short data, and as it is being applied here to each subsequence and averages, we should expect bias in the overall estimate. At best it appears to be a sensitive and hard to interpret measure. 

    \item There is close agreement for entropy estimators using the largest connected component, and induced irreducibility. This similarity is most pronounced for the LAMP models. That suggests that the manner in which we corrected for reducible processes does not matter, but that the LAMP models are also more robust to details such as this.
    
    \item Lower estimated entropies are often indicative of a better model (for the purpose of estimating entropy) that captures more of the predictability of the sequence. The estimates generated using the LAMP models are the lowest estimates for the \textsc{Wikispeedia} and \textsc{Reuters} datasets. The lower entropy estimate values imply that more information about the sequence is captured by the LAMP model than the underlying first-order Markov model. This implies that the model captures more information about the behaviour of the system, so the generated sequence is more predictable. By better capturing the long-term dependency structures in the data, we achieve a more meaningful model. The estimates obtained using the LAMP model appear appropriate for all datasets. This supports the claims of \cite{10.1145/3038912.3052644}, that the LAMP model provides a better, succinct, representation of these datasets, but our claim is slightly stronger in the context of entropy because they compare against $k$-order and adaptive Markov chains. 
    
    \item For the \textsc{Lastfm} and \textsc{BrightKite} datasets, the estimates obtained using the first-order Markov models are lower than those obtained with the LAMP model. These two datasets represent artist sequences on a music streaming service and checked in user locations and it is possible that these sequences can be accurately captured without the LAMP structure. This is somewhat in contrast with \cite{10.1145/3038912.3052644}, but note that this finding is only in respect to the entropy estimates and there are other facets of a model that are also important. 
    
\end{enumerate}

\section{Conclusion}

By deriving a theoretical result for the entropy rate of a linear additive Markov Process, we extend the versatility of LAMPs to information-centic applications and derive a surprising equality between the entropy rate of a LAMP and the underlying first-order Markov Chain. 
A consequence of this is that realisations of processes having different long-range dependency structures -- and sample paths which are therefore quite different from each other -- might nonetheless be indistinguishable based on entropy rates.
This is because the entropy rates depend only on the underlying transition matrix structure, and not the long-range dependence structure.
The LAMP however provides a means to distinguish between such processes, by visualisation of the long-range dependency structure.
We demonstrate this via numerical simulations and use of the Cramer's V statistic.
While LAMPs exhibit long-term dependency structures, their symbol-by-symbol dependency is low due to the randomness introduced by the distribution over the choice of conditioning state.
Linear Additive Markov Processes present an intuitive yet useful extension of traditional Markov chain models to incorporate long-range dependence.
This work demonstrates that LAMPs can be understood within a traditional information-theoretic framework, and many further results are possible in this context.
Future work will explore the application of LAMPs to understanding long-range phenomena such as in social media communications \cite{mathews2017nature},
to provide enhancements to network traffic measurements such as \cite{6242413}, and to derive online estimators for entropy such as in~\cite{776004}.

\bibliographystyle{alpha}
\bibliography{LAMPBib}

\begin{appendices}
\section{On connecting a Markov chain: Choosing artificial states and connection probabilities}
For a unique stationary distribution of a Markov Chain to exist, it is necessary for the Markov Chain to be ergodic \cite{norris1998markov}. By adding an artificial state to the chain, which is connected to all other states with some transition probability  $p_{\text{artificial}}$, we can guarantee the chain is ergodic, and that a unique stationary distribution exists. 
The choice of hyperparameter $p_{\text{artificial}}$ is critical to ensuring this artificial state does not dramatically alter the behaviour of the Markov chain and instead acts as a link between otherwise disconnected communicating classes.

The parameter $p_{\text{artificial}}$ was chosen for each of the four datasets for both the LAMP and first-order Markov Chain approaches by calculating entropy estimates for $p_{\text{artificial}} = 2^{-i}$ for $i = 1,....,25$ for the first-order Markov estimates and $i = 1, ..., 50$ for the LAMP estimates. Since we are selecting for small values of $p_{\text{artificial}}$, it is also important to ensure we avoid numerical precision errors. It is important that the estimated value is robust to small changes in $p_{\text{artificial}}$. By selecting values of $p_{\text{artificial}}$ for which the estimate is stable and appears to have converged, we can estimate the entropy of the process. The true entropy value will differ between models and datasets, so we are not concerned with the value which the estimate converges to, only the stability of the estimate.

Generally, these estimates approach a stable value for larger values of $i$, although it is not known if this trend would continue, due to precision limitations. The largest value for which the estimate is stable was taken to perform each of the estimates was used.

The results of these simulations are shown in Figure \ref{fig:p_sims}. For all dataset model combinations except for the entropy estimates obtained using the first-order Markov model for the \textsc{Wikispeedia} and the \textsc{Reuters} datasets, the entropy values approach a limiting value from above, with most estimates converging to a stable value. The only estimate which does not appear to converge is the estimate obtained with the first-order model for the \textsc{BrightKite} dataset, which begins to flatten between -7 to -10, but then approaches zero. This may be due to numerical limitations, but this behaviour will be explored in future work.

In Figure \ref{fig:p_sims}, the entropy estimates are normalised to have a minimum value of 0 and a maximum value of 1, to enable the shape of the curves to be easily compared. This highlights the shape of the curve and emphasises the direction of convergence. These normalised convergence curves for the \textsc{Wikispeedia} and \textsc{Reuters} datasets overlap for the first-order Markov model, so the values for the \textsc{Reuters} dataset are offset by +0.04 for visualisation. Similarly, for the LAMP model convergence curves, a small offset is added to \textsc{Reuters} and \textsc{BrightKite} datasets in the LAMP model visualisation (+0.01 and -0.005 respectively).

The choice of $p_{\text{artificial}} = 2^{-15}$ appears appropriate for all datasets where the estimates appear to converge, while remaining large enough to avoid precision errors. This artificial transition probability is equivalent to observing a single transition amongst 32 768, which is sufficiently small to have a minimal impact on the scales we are considering. It is important that this choice is evaluated for different datasets, where the size of the data may make a different value more suitable.

\begin{figure}[t!]
    \centering
    \includegraphics[width=0.43\textwidth]{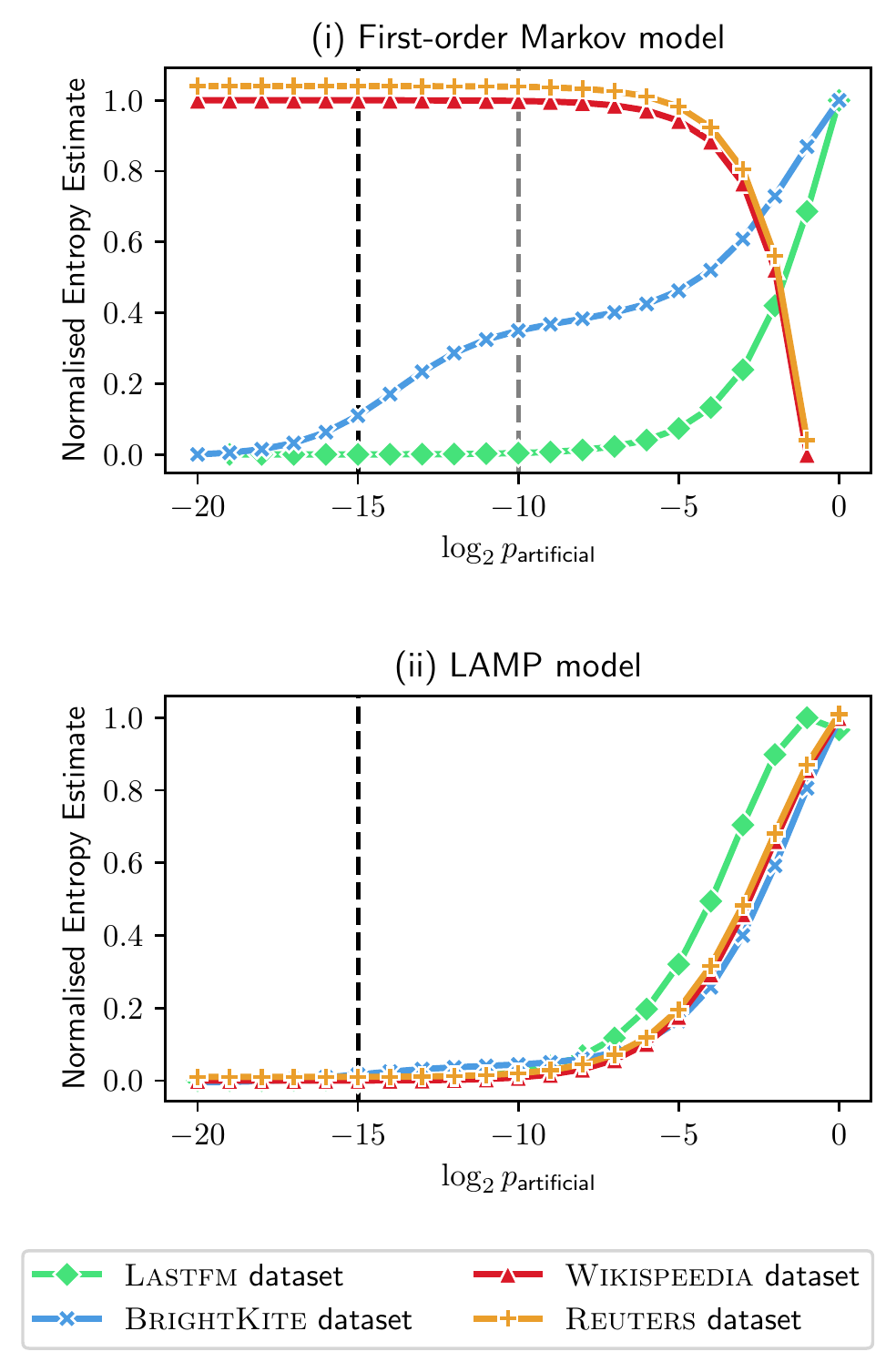}
    \caption{Plots to show convergence for the normalised entropy estimate value against $\log_2 p_{\text{artificial}}$, which was used to ensure ergodicity. (i) Shows hyperparameter sensitivity for the first-order Markov models and (ii) for the LAMP models. Each plot shows the effect of the weight of this artificial link on both the first-order Markov model estimate and the estimate obtained using the LAMP model for each dataset. We aim to find a region where the estimates are insensitive to the choice of hyperparameter. A dashed black line on each plot indicates the value when the artificial link weight is $2^{-15}$. This value was chosen as a global value, since it is a reasonable choice for all dataset model combinations, apart from the BrightKite dataset first-order Markov model, when a value of $2^{-10}$ was used to obtain the final estimate. This alternative value is indicated by a grey dashed line. The convergence curve for the \textsc{Wikispeedia} and \textsc{Reuters} datasets overlapped for the first-order Markov model, so the values for the \textsc{Reuters} dataset was offset by +0.04 for visualisation. Small vertical offset was also added to the \textsc{Reuters} and \textsc{BrightKite} datasets in the LAMP model visualisation (+0.01 and -0.005 respectively).}
    \label{fig:p_sims}
\end{figure}
\end{appendices}

\newpage

\end{document}